\documentclass[11pt,letter]{scrartcl}

\usepackage{geometry}
\usepackage{times,mathptmx}
\usepackage{color}
\usepackage{graphicx}

\usepackage{tabularx}
\usepackage{colortbl}
\usepackage{ifpdf}
\usepackage{booktabs}
\usepackage{amsmath, amsthm, amssymb}
\usepackage[nothing]{algorithm}
\usepackage[noend]{algorithmic}
\usepackage{subfigure}	%Bilder nebeneinander und verschachtelt

\usepackage{tikz}
\usetikzlibrary{arrows}

\definecolor{link}{rgb}{0,0,0.5}
\usepackage[plainpages=false,pdfpagelabels,colorlinks,linkcolor=link,citecolor=link,bookmarks,bookmarksopen,bookmarksopenlevel=2]{hyperref}

\usepackage{placeins}

\usepackage{xspace}
\usepackage{mathtools}
\usepackage{xfrac}
\usepackage{enumerate}

\newtheorem{theorem}{Theorem}
\newtheorem{lemma}{Lemma}
\newtheorem{corollary}{Corollary}
\newtheorem{proposition}{Proposition}
\newtheorem{fact}{Fact}

\newfont{\classFont}{cmsy10}

\newcommand{\classNP}{\text{\fontencoding{OMS}\fontfamily{cmsy}\selectfont NP}}
\newcommand{\classDTIME}{\text{\fontencoding{OMS}\fontfamily{cmsy}\selectfont DTIME}}

\newcommand{\R}{\mathbb R}

\newcommand{\bigO}[1]{O (#1)}
\newcommand{\bigOmega}[1]{\Omega(#1)}

\newcommand{\Congest}{\text{\fontencoding{OMS}\fontfamily{cmsy}\selectfont CONGEST}}

\newcommand{\connected}{\emph{connected}}
\newcommand{\notconnected}{\emph{not connected}}

\newcommand{\paid}{\emph{currently-paid}}
\newcommand{\open}{\text{open}}
\newcommand{\closed}{closed}
\newcommand{\greedy}{\text{\emph{GreedyFL}}\xspace}
\newcommand{\FRDist}{\text{\emph{FacRev-Dist}}\xspace}
\newcommand{\FRJain}{\text{\emph{FacRev-Jain}}\xspace}

\newcommand{\receive}{\ensuremath{\mbox{\sc  Receive}}}
\newcommand{\send}{\ensuremath{\mbox{\sc  Send}}}

\newcommand{\ie}{i.e\text{.}\xspace}
\newcommand{\abs}[1]{\left\lvert#1\right\rvert}
\newcommand{\WLOG}{w.l.o.g\text{.}\xspace}
\newcommand{\resp}{resp\text{.}\xspace}

\newcommand{\et}{et al\text{.}\xspace}
\newcommand{\eg}{e.g\text{.}\xspace}

\newcommand{\ouralg}{\textsc{FacilitySelect}\xspace}

\title{A Distributed Approximation Algorithm for the Metric Uncapacitated Facility Location Problem in the $\Congest$ Model\thanks{Partially supported by the EU within FP7-ICT-2007-1 under contract no. 215270 (FRONTS),  DFG-project ``Smart Teams'' within the SPP 1183 ``Organic Computing'', and the Paderborn Institute for Scientiﬁc Computation (PaSCo).}}

\author{Patrick Briest \quad Bastian Degener \quad Barbara Kempkes\\ Peter Kling \quad Peter Pietrzyk \\[2mm]
\small Heinz Nixdorf Institute, Computer Science Department,\\[-2mm]
\small University of Paderborn, 33095 Paderborn, Germany\\
\small \textsf{patrick.briest@upb.de, bastian.degener@upb.de, barbaras@upb.de, }\\[-2mm] \small \textsf{peter.kling@upb.de, peter.pietrzyk@upb.de}
}

\begin{document}

\date{}
\maketitle \thispagestyle{empty}

\vspace{-1cm}

\begin{abstract}
We present a randomized distributed approximation algorithm for the metric uncapacitated facility location problem. The algorithm is executed on a bipartite graph in the $\Congest$ model yielding a $(1.861 + \epsilon)$ approximation factor, where $\epsilon$ is an arbitrary small positive constant.  It needs $\bigO{n^{3/4}\log_{1+\epsilon}^2(n)}$ communication rounds with high probability ($n$ denoting the number of facilities and clients). To the best of our knowledge, our algorithm currently has the best approximation factor for the facility location problem in a distributed setting. It is based on a greedy sequential approximation algorithm by Jain \et (J. ACM 50(6), pages: 795-824, 2003). The main difficulty in executing this sequential algorithm lies in dealing with situations, where multiple facilities are eligible for opening, but (in order to preserve the approximation factor of the sequential algorithm) only a subset of them can actually be opened. Note that while the presented runtime bound of our algorithm is ``with high probability'', the approximation factor is not ``in expectation'' but always guaranteed to be $(1.861 + \epsilon)$. Thus, our main contribution is a sublinear time selection mechanism that, while increasing the approximation factor by an arbitrary small additive term, allows us to decide which of the eligible facilities to open.
\end{abstract}

\newpage

\section{Introduction}
Facility location is one of the most studied optimization problems in operations research and captures a large variety of applications. A classical motivation is placing facilities (\eg, warehouses) in such a way that the combined costs of customer satisfaction and warehouse construction are minimized. However, there are also plenty of applications in distributed scenarios. For instance, in wireless networks  a set of  nodes has to be chosen to provide some services (\eg, a distributed database). Making such services available incurs costs at those \emph{facility} nodes, while all remaining nodes act as \emph{clients}. They use the services of the nearest facility node, and are charged a cost proportional to the corresponding distance. The objective is to determine a set of facility nodes such that the costs caused by the facilities and the clients is as low as possible.

\paragraph{Formal problem definition.} We consider the \emph{metric uncapacitated Facility Location} problem in a distributed setting. Here, we are given a complete bipartite graph $G =F \cup C$ consisting of a set of \emph{facilities} $F$ and a set of \emph{clients} $C$. To each facility $i \in F$ an opening cost $f_i \in \R_{\geq 0}$ is assigned. Each edge $\{i,j\}$ in $G$ is weighted with the value $c_{ij} \in \R_{\geq 0}$ that represents the costs of connecting client $j$ with facility $i$. The objective is to determine a subset of the facilities to be opened and connect every client to at least one open facility in such a way that the sum of the opening costs and connection costs is minimized. The linear program representation of the facility location problem and its dual program are as follows:
%\[
	%\begin{array}{rrcllllrrclll}
												%\multicolumn{6}{l}{ \textrm{{\bf Facility Location IP}}}								 																												 & \quad & \multicolumn{5}{l}{ \textrm{{\bf Dual of the Facility Location LP}}}\\
		%\mbox{minimize}		& \multicolumn{5}{l}{ \displaystyle \sum_{i \in F} f_i y_i + \sum_{i \in F, j\in C}  c_{ij}x_{ij}} & \quad & \mbox{maximize}& \multicolumn{5}{l}{ \displaystyle \sum_{j \in C} \alpha_j}\\
		%
		%\mbox{subject to}	&  \sum_{i \in F} x_{ij}	&\geq	& 1 		&  j \in C								& (1) & \quad\quad & \mbox{subject to}	&  \alpha_j -\beta_{ij}	&\leq	& c_{ij} 		&i \in F, j \in C		& (3) \\
					%& \displaystyle y_i-x_{ij}				&\geq	& 0 		&  i \in F,  j \in C 	      & (2) & \quad\quad &                   &  \sum_{j \in C} \beta_{ij}				&\leq	& f_i 		& i \in F  	& (4)\\
					%& \displaystyle x_{ij}					&\in	& \{0,1\} 	&  i \in F, j \in C 	    &     & \quad\quad &                   & \displaystyle \beta_{ij}					&\geq & 0 	& i \in F, j \in C 	&\\
					%& \displaystyle y_i						&\in	& \{0,1\} 	& i \in F								                &     & \quad\quad &                   & \displaystyle \alpha_i						&\geq & 0 	& j \in C								&\\
	%\end{array} 
%\]

\[
	\begin{array}{rrclll}
												\multicolumn{6}{c}{ \textrm{{\bf Facility Location IP}}}								 																												 \\
		\mbox{minimize}		& \multicolumn{5}{l}{ \displaystyle \sum_{i \in F} f_i y_i + \sum_{i \in F, j\in C}  c_{ij}x_{ij}} \\
		
		\mbox{subject to}	&  \sum_{i \in F} x_{ij}	&\geq	& 1 		&  j \in C								& (1) \\
					& \displaystyle y_i-x_{ij}				&\geq	& 0 		&  i \in F,  j \in C 	      & (2) \\
					& \displaystyle x_{ij}					&\in	& \{0,1\} 	&  i \in F, j \in C 	    &     \\
					& \displaystyle y_i						&\in	& \{0,1\} 	& i \in F								                  \\
	\end{array} 
\]
\[
	\begin{array}{rrclll}
											\multicolumn{5}{c}{ \textrm{{\bf Dual of the Facility Location LP}}}\\
		\mbox{maximize}& \multicolumn{5}{l}{ \displaystyle \sum_{j \in C} \alpha_j}\\
		\mbox{subject to}	&  \alpha_j -\beta_{ij}	&\leq	& c_{ij} 		&i \in F, j \in C		& (3) \\
                  &  \sum_{j \in C} \beta_{ij}				&\leq	& f_i 		& i \in F  	& (4)\\
   & \displaystyle \beta_{ij}					&\geq & 0 	& i \in F, j \in C 	&\\
                & \displaystyle \alpha_i						&\geq & 0 	& j \in C								&\\
	\end{array} 
\]
The variable $y_i$ indicates whether facility $i$ is open ($y_i = 1$) or closed ($y_i = 0$). The other indicator variable $x_{ij}$ has the value $1$ if the client $j$ is connected to facility $i$, and $0$ otherwise. The constraints ($1$) guarantee that each client is connected to at least one facility, while the constraints ($2$) make sure that a client can only be connected to an open facility. The problem we consider is \emph{metric}, since the values $c_{ij}$ are required to satisfy the triangle inequality (\ie, $\forall i,j,i',j': c_{ij} \leq c_{ij'}+c_{j'i'}+c_{i'j}$), and \emph{uncapacitated}, since an arbitrary number of clients can be connected to an open facility. Furthermore, for the sake of presentation we assume all $f_i,c_{ij}$ to be normalized such that the smallest non-zero value is $1$. The dual program will be used in the description of our approximation algorithm. Intuitively, a $\alpha_j$ variable can be seen as the amount the client $j$ is willing to pay for being connected to a facility. From the point of view of a facility $i$, $\alpha_j$ is the sum of $c_{ij}$ (the amount $j$ pays for a connection to $i$) and $\beta_{ij}$ (the amount it pays for opening $i$).

Our algorithm is executed in the $\Congest$ model, which was introduced by Peleg \cite{Pe00} and is commonly used to model the execution of distributed algorithms on graphs: Algorithms are executed in synchronous send-receive-compute rounds. In a single round, each node sends a message to each of its neighbors in the graph. Note that the messages sent to each neighbor by a single node are not required to contain the same information. Once all nodes have sent their messages, they receive a single message from each of their neighbors. After all the messages have been received, every node is allowed to spend an arbitrary amount of time for computation (\ie, computation is for free and we are only interested in the number of communication rounds). The end of the computation by all nodes marks the start of a new send-receive-compute round. The message size in the $\Congest$ model is bounded. We will limit the size of the messages used in our algorithm to $\bigO{\log(n)}$ bits. The same limitations were also used in \cite{FLapprox}, \cite{FLunfiformApprox} and \cite{FLdistPrimalDual}. This bound is reasonable, because it allows the nodes to send their ID in a single message. Due to this constraint on message size, we restrict the values $f_i$ and $c_{ij}$ in such a way that it is possible to represent them with $\bigO{\log(n)}$ bits (i.e. to be able to send them in a single message). 

The graph our algorithm is executed on is the complete bipartite graph of clients and facilities. This means that within a single round each client can communicate with all facilities and each facility can communicate with all clients. Since nodes of the same partition can not communicate directly with each other, but instead have to use nodes of the other partition as relays, gathering all information about the graph requires $\bigOmega{n}$ rounds. Note that without a limit on the message size our problem could be trivially solved in four rounds (all information about the graph is gathered at a single node in two rounds, an optimal solution is computed and distributed in two more rounds). We want to stress that, although in our algorithm the nodes communicate with all their neighbors in each round, it is possible to restate the algorithm such that nodes $i$ and $j$ where $c_{ij}$ is ``large'' (more than $\max_{j \in C}(\min_{i \in F}(c_{ij}+f_i))$) never communicate with each other. 

\paragraph{Our contribution.} We present a distributed approximation algorithm for the metric uncapacitated facility location problem. It is based on a greedy algorithm by Jain \et \cite{FLgreedy} (from here on referred to as \greedy) and yields a guaranteed (not in expectation) $(1.861+\epsilon)$-approximation in $\bigO{n^{3/4}\log_{1+\epsilon}^2(n)}$ rounds (with high probability). Our algorithm is executed in the $\Congest$ model on a complete bipartite graph. Although it is strongly related to the \emph{GreedyFL} algorithm, there are (due to the parallel execution) new challenges regarding the selection process that occurs when multiple facilities are eligible for opening. The difficulty is that we require the cost of the solution computed by our algorithm to be always at most $(1.861+\epsilon)$ times worse than an optimal solution (in contrast to similar work by Blelloch \et \cite{FLparallel}, where the approximation factor is increased by factor $2$). Moreover, unlike other work (\eg, \cite{podc2010}, where the approximation factor is given ``in expectation''), our algorithm provides a worst-case guarantee on the approximation factor.

\paragraph{Related work.} During the last two decades, the uncapacitated metric facility location problem was of great interest, and a lot of progress has been made concerning the running time and approximation factor of sequential algorithms solving it. Aardal \et introduced the first polynomial time algorithm yielding a $3.16$-approximation \cite{FLfirstApprox}. Improving this approximation factor was the topic of a multitude of research papers. For example, Chudak \et improved the approximation factor to $(1+2/e) \approx 1.74$ \cite{FLLP2}. The factor has been improved, until eventually Byrka designed a $1.5$-approximation algorithm \cite{FLbestKnown}, which, at this point in time, yields the best known approximation. Although both last mentioned algorithms yield very good approximation factors, they have, due to applying LP-Rounding, high running times. Thus, the design of algorithms with slightly worse approximation factors, but better running times, was also of interest in the past.

Jain and Vazirani used the primal-dual approach to develop and analyze an algorithm with approximation factor $3$ and a running time of $O(n^2 \log(n))$ \cite{FLprimalDual}. A simplified and faster $\bigO{n^2}$ version of this algorithm was introduced by Mettu and Plaxton \cite{FLmettuPlaxton}. Later on, Jain \et improved the former results by presenting two algorithms in \cite{FLgreedy}: One with running time $\bigO{m\log(m)}$ and an approximation factor of $1.861$ (which we refer to as \greedy) and another one with running time $\bigO{n^3}$ and an approximation factor of $1.61$ ($n$ denoting the number of nodes and $m$ the number of edges in the complete bipartite graph of facilities and clients). Later on, building upon \cite{FLgreedy}, Mahdian \et improved the factor to $1.52$ \cite{FLMahdian152}.

 Under the assumption that $\classNP \not \subseteq \classDTIME(n^{\log(\log(n))})$, Guha \et showed in \cite{FLhardness} that no polynomial time algorithm with an approximation factor better than $\approx 1.463$ exists. Thus, an improvement of either the lower or the upper bound for the approximation factor, even in the $3$rd or higher fractional digit, is of great interest.

The following results concerning the facility location problem can be found in the distributed scenario: In \cite{FLdistPrimalDual} Pandit \et present an algorithm yielding a $7$-approximation and a running time of $\bigO{\log(n)}$. Their algorithm is a parallel version of the primal-dual algorithm by Jain \et \cite{FLprimalDual} and is -- like our algorithm -- executed in the $\Congest$ model on a complete bipartite graph with message size limited to $\bigO{\log(n)}$ bits. Pandit \et assume that the difference between the opening costs of the cheapest and the most expensive facility can be arbitrary large. It is reasonable to drop this assumption, since they (as well as we) require that the facility opening costs and distances between clients and facilities are encoded with $\bigO{\log(n)}$ bits. By dropping this assumption and modifying the algorithm and analysis by Pandit \et in a small way (changing a factor from $2$ to $(1+\epsilon)$), it is possible to achieve a $(3+\epsilon)$ approximation factor in $\bigO{\log_{1+\epsilon}(n)}$ rounds. Further improvement of the factor with this approach is not possible, since the approximation factor of \cite{FLprimalDual}, which Pandit \et parallelized, is $3$. This means that their algorithm, while faster than ours, has a worse approximation factor.

Recently, in \cite{podc2010} Pandit \et all presented a technique that can be used to execute greedy facility location algorithms (like \greedy) in parallel in polylogarithmic time. Although they consider a similar problem as we do, their results are quite different from ours: Their approximation factor is $\bigO{1}$ in expectation (their algorithm can produce an arbitrary bad solution, even though this is very unlikely), while we can guarantee an approximation factor of $(1.861+\epsilon)$ in the worst case. Also, although they do not state their exact approximation factor, the expected approximation factor achieved with their technique cannot, to the best of our knowledge, be decreased below $4c$, where $c$ is the approximation factor of the sequential algorithm (\ie, $\approx 7.444$ if used with \greedy).

A similar result to \cite{podc2010} was presented by Blelloch \et in \cite{FLparallel}. Instead of the $\Congest$ they use a PRAM model and achieve a $(3.722+\epsilon)$ approximation with running time $\bigO{\log_{1+\epsilon}^2(n)}$ by parallelizing the \greedy algorithm. In order to choose which facilities to open, they use a technique introduced by Rajagopalan and Vazirani \cite{symmetry98}. This technique can also be used to execute greedy facility location algorithms (in a distributed manner and polylogarithmic time) and yields a guaranteed approximation factor of $2c$, where $c$ is again the approximation factor of the corresponding sequential algorithm (\ie $\approx 3.722$  if used with \greedy).

Other results that also use the bound of $\bigO{\log(n)}$ on the message complexity are \cite{FLapprox} and \cite{FLunfiformApprox}. In \cite{FLapprox} Moscibroda \et show that in $\bigO{k^2}$ communication rounds a $\bigO{k(n_F\rho)^{1/k}\log(n_F+n_C)}$ approximation in $\bigO{k^2}$ communication rounds ($n_F$ and $n_C$ are the number of facilities, respectably clients, and $\rho$ a coefficient dependent on instance parameters) can be achieved in the more general \emph{non-metric} facility location problem. Gehweiler \et show in \cite{FLunfiformApprox} that $3$ rounds are sufficient to compute a $\bigO{1}$-approximation for the uniform facility location problem (opening each facility incurs the same costs). \cite{FLunfiformApprox} applies the approach of Mettu and Plaxton \cite{FLmettuPlaxton}, which has been successful in a lot of other settings as well: In the kinetic setting \cite{KFL}, in game theoretic settings \cite{groupstrategyproof}, for algorithms working in sublinear time \cite{FLSublinearTime}, and when confronted with perpetual changes to the problem instance \cite{FLipdps}.

\paragraph{Structure of the Paper.} In Section 2 we give a detailed description of our approximation algorithm and prove its approximation factor by generalizing the techniques and results of \cite{FLgreedy}. Our main contribution is Section 3. Here, we deal with the problem of selecting facilities to be opened that arises due to the parallel execution of the \emph{GreedyFL} algorithm by Jain \et and show that our distributed algorithm presented in Section 2 terminates in sublinear time.     
 
\section{The distributed Approximation Algorithm} Our parallel algorithm is based on the \emph{GreedyFL} algorithm by Jain \et \cite{FLgreedy}. As stated above, it is executed in the $\Congest$ model on a complete bipartite graph. We will first describe the algorithm and then give an analysis of its approximation factor and running time. 

\paragraph{Algorithm description.} Due to the completeness of the bipartite graph, in a single synchronous round each client can send to and receive a message from all facilities. The same applies analogously to each of the facilities. During the algorithm's execution, clients and facilities can be in various states: clients can be \emph{not connected} or \emph{connected}, while facilities can be \emph{closed}, \emph{currently-paid}, or \emph{open}. Intuitively, \emph{currently-paid} facilities compete with each other for the permission to change their state to \emph{open}. All clients start in the \emph{not connected} state and all facilities in the \emph{closed} state. When a client $j$ changes its state to \emph{connected}, it is assigned to an \emph{open} facility $i$ ($\varphi (j) := i$) and never changes its state again. A \emph{closed} facility can only change its state to \emph{currently-paid}, while a \emph{currently-paid} facility can either become \emph{closed} or $\emph{open}$. Once a facility is $\emph{open}$, it never changes its state. The algorithm terminates, once all clients are \emph{connected}. The solution returned by our algorithm is the set of all \emph{open} facilities and the assignment $\varphi$. Each facility knows now whether it is open or not and each client knows the facility it is assigned to.

\begin{algorithm}[ht!]
\caption{\textsc{Executed by client $j$}}
\label{alg:client}
\begin{algorithmic}[1]
\STATE $\forall i \in F: \quad \send[\alpha_j, status(j)]$\\
\STATE Execute Algorithm 4 until there are no \emph{currently-paid} facilities left
\IF {$status(j) = \notconnected$}
	\STATE Compute $\alpha_j := (1+\epsilon) \alpha_j$
\ENDIF	
\end{algorithmic}
\end{algorithm}

The algorithm is executed in a loop and each passing of the loop is referred to as a \emph{phase}. The behavior of a client during a phase is described in Algorithm 1, while the behavior of a facility is described in Algorithm 2. Facilities and clients interact with each other by transmitting messages; to each \textsc{Send} operation there exists a corresponding \textsc{Receive} operation.

The idea of the algorithm is to assign to each client $j$ a variable $\alpha_j$ initialized with $1$. This variable represents the amount a client is willing to pay for being connected to a facility and for opening this facility in the current phase (\ie, the $\alpha_j$ variable of the dual program). Each \emph{not connected} client $j$ increases this variable in every phase by multiplying it with $(1+\epsilon)$ where $\epsilon$ is an arbitrary small constant greater than $0$ (fixed at the start of our algorithm). After $\alpha_j$ is increased, all clients send their current $\alpha_j$ to all facilities. 

\begin{algorithm}[ht!]
\caption{\textsc{Executed by facility $i$}}
\label{alg:facility}
\begin{algorithmic}[2]
\STATE $\forall j \in C: \quad  \receive[\alpha_j,status(j)]$
\IF {($status(i) = \closed$)}
	\STATE $U := \{j|j \in C \wedge status(j) = \notconnected \}$
	\STATE Compute $coveredCost := \sum_{j \in U} \max (0, \alpha_j - c_{ij})$%
	\IF {($f_i \leq coveredCost$)}
			\STATE $status(i) :=$ \paid \\
	\ENDIF
\ENDIF
\STATE Execute Algorithm 3 until there are no \emph{currently-paid} facilities left
\end{algorithmic}
\end{algorithm}

Let $U$ be the set of all \emph{not connected} clients and $\beta_{ij} := \max(\alpha_j-c_{ij}, 0)$. We say that client $j$ \emph{contributes} to facility $i$, if $\beta_{ij} >0$ (the client pays $c_{ij}$ for being connected to and $\beta_{ij}$ for the opening of facility $i$). Upon receiving the $\alpha_j$ variables, every \emph{closed} facility $i$ computes $\sum_{j \in U} \beta_{ij}$, which represents the amount \emph{not connected} clients are willing to contribute to the payment of opening facility $i$. If the contribution to the opening of facility $i$ (i.e. $\sum_{j \in U} \beta_{ij}$) reaches a point such that $f_i \leq \sum_{j \in U} \beta_{ij}$, it changes its status to \emph{currently-paid}.

\begin{algorithm}[ht!]
\caption{\textsc{Executed by facility $i$ }}
\label{alg:facility2}
\begin{algorithmic}[2]
\STATE $R_i :=$ uniformly distributed random number in $[0,1]$
\STATE $\forall j \in C: \quad \send[R_i, status(i)]$\\
\STATE $\forall j \in C: \quad \receive[S_j, status(j)]$
\STATE Define $T_i := \max_I(S_j)$ where $I := \{j| j \in C, status(j) = \notconnected, \alpha_j \geq c_{ij} \}$\\
\STATE $thisRoundOpened_i := false$
\IF {($T_i = R_i \wedge status(i) = \paid $)}
	\STATE $status(i) := \open$
	\STATE $thisRoundOpened_i := true$	
\ENDIF
\STATE $\forall j \in C: \quad \send[status(i), thisRoundOpened_i]$\\
\STATE $\forall j \in C: \quad \receive[status(j)]$\\
\STATE $I' := \{j| j \in C, status(j) = \notconnected, \alpha_j \geq c_{ij} \}$
\IF{$(status(i) = \paid \wedge \sum_{j \in I'} \max (0, \alpha_j - c_{ij}) < f_i)$}
	\STATE $status(i) = \closed$
\ENDIF
\end{algorithmic}
\end{algorithm}

Opening all the \emph{currently-paid} facilities could result in a very bad approximation, since it is possible that a single \emph{not connected} client $j$ is contributing to an arbitrary number of \emph{currently-paid} facilities (i.e. $\beta_{i,j},\beta_{i',j} > 0, i\neq i'$) and opening just one of these facilities would be sufficient. Thus, a selection procedure is needed to determine which \emph{currently-paid} facilities are going to be permanently opened (change their status to \emph{open}) and which of them will remain closed in this phase (changing their status from \emph{currently-paid} to \emph{closed}). Furthermore, it is required to determine which client is connected to which facility. All of this is achieved by the interaction of Algorithm 3 (from the point of view of facility $i$) and Algorithm 4 (from the point of view of client $j$). In each iteration of Algorithm 3 and 4, a subset of \emph{currently-paid} facilities is chosen to be opened. The development and analysis of these two algorithms is the main contribution of our paper, as they are responsible for selecting the facilities (which are going to be opened) in a fast and efficient way. The algorithms operate by generating a random number at each facility and then opening a facility, if its number is the highest one in its neighborhood (the neighborhood of a facility $i$ consists of all facilities $i'$, such that there is at least one client that contributes to $i$ and $i'$). Only $\bigO{\log(n)}$ bits are used to represent the generated random number, which is sufficient to guarantee that, with high probability, no two facilities generate the same number. The main objective of Algorithm 3 and 4 is to establish the following fact:

\begin{fact}
 In every phase, after the iterative execution of Algorithm 3 and 4 terminates, the following is true for all facilities:
\begin{enumerate}
	\item[(i)] No facility is \emph{currently-paid}.
	\item[(ii)] If facility $i$ is \emph{open}, then $\sum_{j \in U'} \beta_{ij} \geq f_i$, where $U' = \{j|\varphi(j) = i\}$.
	\item[(iii)] If facility $i$ is \emph{closed}, then $\sum_{j \in U} \beta_{ij} < f_i$, where $U = \{j| state(j) = \notconnected\}$.
\end{enumerate}
\end{fact}

Intuitively, (ii) means that a facility is only opened, if all clients connected to it can (together) pay for its opening, and (iii) means that, after a phase ended, there is no \emph{closed} facility to which clients contribute enough to pay for its opening. The fact that a \emph{closed} facility with  $\sum_{j \in U} \beta_{ij} \geq f_i$ cannot enter the next phase, as implied by (iii), is of importance. It ensures that the desired approximation factor can be guaranteed, which is curcial for (the later introduced) Lemma 2.

%As long as there is at least one \emph{currently-paid} facility left, Algorithm 3 and 4 are executed. For ease of presentation, the case in which a \emph{not connected} client starts contributing to facilities opened at least one phase earlier is not considered in Algorithms 3 and 4. If such an event occurs, the client is connected to any arbitrary facility it contributes to.  

\begin{algorithm}[ht!]
\caption{\textsc{Executed by client $j$ }}
\label{alg:client2}
\begin{algorithmic}[1]
\STATE $\forall i \in F: \quad \receive[R_i, status(i)]$\\
\STATE Define $S_j := \max_J(R_i)$ where $J := \{i| i \in F, status(i) = \paid, \alpha_j \geq c_{ij} \}$\\
\STATE $\forall i \in F: \quad \send[S_j, status(j)]$\\
\STATE $\forall i \in F: \quad \receive[status(i),thisRoundOpened_i]$\\
\STATE $J' := \{i| i \in F, status(i) = \open, \alpha_j \geq c_{ij}, thisRoundOpened_i = true \}$\\
\IF{$(status(j) = \notconnected \wedge J' \neq \emptyset)$}
	\STATE $status(j) =$ \connected
	\COMMENT{There is exactly $1$ element in $J'$, see Algorithm 3}
	\STATE $\varphi (j) = i$ $( i \in J' )$
\ENDIF
\STATE $\forall i \in F: \quad \send[status(j)]$\\	

\end{algorithmic}
\end{algorithm}

 Note that after the algorithm's termination the cost of the solution computed by our algorithm lies in the interval $[\sum_{j \in C} \alpha_j/(1+\epsilon),\sum_{j \in C} \alpha_j]$. This is due to the following facts: Each client is connected to exactly one facility and this facility is \emph{open}. Since $\beta_{ij} = \max(\alpha_j-c_{ij}, 0)$ a client $j$ contributes to the opening of a facility $i$ only after $\alpha_j$ covered the costs $c_{ij}$ of connecting $i$ to $j$. Also, since we require that only after $f_i \leq \sum_{j \in U} \beta_{ij}$ facility $i$ changes its state to \emph{open}, the opening of $i$ is completely paid for and since $U$ is the set of all \emph{not connected} clients, it is made sure that a client contributes to the opening of at most one facility.

The running time of the entire algorithm is bounded by $\bigO{n^{3/4}\log_{1+\epsilon}^2(n)}$ rounds: The number of phases is  $\bigO{\log_{1+\epsilon}(n)}$, since in the worst case a client has distance $\bigO{poly(n)}$ to its nearest facility and this facility has opening costs of $\bigO{poly(n)}$ (the $c_{ij}$ and $f_i$ are limited by $\bigO{\log(n)}$ bits). Thus, increasing the $\alpha_j \geq 1$ value iteratively by multiplying it with $(1+\epsilon)$, where $\epsilon > 0$, results in $\bigO{\log_{1+\epsilon}(n)}$ rounds. Later, we will show that with high probability Algorithm 3 and 4 terminate after $\bigO{n^{3/4}\log(n)}$ rounds. The value chosen for $\epsilon$ not only determines the runtime, but also the approximation factor, which is $1.861(1+\epsilon)^2$.

\paragraph{Analysis of the approximation factor and runtime.} Intuitively, the \greedy algorithm can be seen as continuously and simultaneously increasing the $\alpha_j$ values. This allows \greedy to open a facility $i$ in the exact moment it is paid for (i.e., when $\sum_{j \in U} \beta_{ij} = f_i$); a property which is essential for proving its approximation factor. In contrast, our algorithm increases the $\alpha_j$ values in discrete steps and thus, in general, it may occur that $\sum_{j \in U} \beta_{ij} \geq f_i$. The goal of this section is to show that -- even though our algorithm may yield a very different solution from the sequential algorithm -- the costs of the optimal and our solution do not differ too much. Namely, we show that our  approximation factor is at most $(1+\varepsilon)^2$ times greater than the approximation factor of \greedy. Our proof uses similar arguments as, and is based on, the one of Jain \et, but it turns out that in order to apply the results of Jain \et, we need to get a deeper insight into the structural properties of the linear programs involved.

The $\alpha_j$ and $\beta_{ij}$ values computed by our algorithm can be interpreted as a solution to the dual program of the facility location program presented in the introduction. However, they do not necessarily represent a feasible solution: our algorithm opens a facility when $\sum_{j \in U} \beta_{ij}	\geq f_i$ (summing over all \emph{not connected} clients) and not, as required by constraint (4), over all the clients. If we choose $\gamma$ such that $\alpha^*_j \coloneqq \frac{\alpha_j}{\gamma}$ and $\beta^*_{ij} \coloneqq \alpha^*_j - c_{ij}$ constitute a feasible dual solution, the Duality Theorem yields that our solution is by a factor of at most $\gamma$ more expensive than an optimal solution. For their original, sequential algorithm Jain \et gave an upper bound of $1.861$ for $\gamma$ \cite{FLgreedy}. They denoted the used technique as dual fitting with a factor-revealing LP. We will construct such a factor revealing program (see \FRDist) for our algorithm  and establish a structural connection between it and the LP by Jain \et (see \FRJain) for their sequential algorithm (\greedy). 

To this end, consider a single facility $i$ and assume, that \WLOG $\alpha_j \geq \gamma c_{ij}$ holds only for the first $k$ clients and that $\alpha_1 \leq \alpha_2 \leq \ldots \leq \alpha_k$. The next two lemmas express the constraints imposed by both the algorithm and the facility location problem. They resemble similar lemmas to those used in the proof for the sequential algorithm (cf.~\cite{FLgreedy}) (their proofs can be found in the appendix).

\begin{lemma}\label{lem:metric_lp_property}
	For any two clients $j$, $j'$ and a facility $i$, we have $\frac{\alpha_j}{(1+\epsilon)} \leq \alpha_{j'} + c_{ij'} + c_{ij}$.
\end{lemma}

\begin{lemma}\label{lem:facilities_not_overpaied_property}
	For every client $j$ and facility $i$, the execution of the algorithm guarantees that $\sum_{l=j}^{k} \max(\alpha_j - (1+\epsilon)c_{il}, 0)\leq (1+\epsilon)f_i$. 
\end{lemma}

Our algorithm ensures that the constraints stated by the two lemmas above are met. Thus, we now want to find the minimum value for $\gamma$ such that $\sum_{j = 1}^{k} \max(\frac{\alpha_j}{\gamma} - c_{ij}, 0) \leq f_i$ holds, without violating the lemmas' statements. Equivalently, we want to maximize the ratio $(\sum_{j = 1}^{k} \alpha_j)/(f_i + \sum_{j=1}^{k} c_{ij})$ (again, without violating the constraints). This formulation yields the following family of linear\footnote{Note that this formulation is not a \emph{linear} program, but can be transformed easily into one.} programs, referred to as the factor-revealing LP.

%\[
	%\begin{array}{rrclll}
		%{\textrm{{\bf FacRev-Jain:}}}
		%& \mbox{maximize}		 & \multicolumn{4}{l}{  z_k = \frac{\sum_{j=1}^{k} \alpha_j} {f_i + \sum_{j = 1}^{k}  c_{ij}}}\\
		%
		%& \mbox{subject to} \quad \displaystyle \alpha_j	&\leq	& \alpha_{j+1} 		& \forall j \in \{1, \ldots, k-1\}								& (1^*) \\
					%& \displaystyle \alpha_j				&\leq	& (1+\epsilon)(\alpha_l + c_{ij} + c_{il})		& \forall j,l \in \{1, \ldots, k\} 	& (2^*)\\
					%&  \sum_{l=j}^{k} \max(\alpha_j-(1+\epsilon)c_{il},0)	&\leq	&	(1+\epsilon)f_i			& \forall j \in \{1, \ldots, k\}& (3^*) 	\\
					%& \displaystyle \alpha_j,c_{ij}, f_i	&\geq&					0 & \forall j  \in \{1, \ldots, k\}								& (4^*)\\
	%\end{array} 
%\]
%
%\[
	%\begin{array}{rrclll}
		%{\textrm{{\bf FacRev-Dist:}}}
		%& \mbox{maximize}		 & \multicolumn{4}{l}{  z_k = \frac{\sum_{j=1}^{k} \alpha_j} {f_i + \sum_{j = 1}^{k}  c_{ij}}}\\
		%
		%& \mbox{subject to} \quad \displaystyle \alpha_j	&\leq	& \alpha_{j+1} 		& \forall j \in \{1, \ldots, k-1\}								& (1^{**}) \\
					%& \displaystyle \alpha_j				&\leq	& (1+\epsilon)(\alpha_l + c_{ij} + c_{il})		& \forall j,l \in \{1, \ldots, k\} 	& (2^{**})\\
					%&  \sum_{l=j}^{k} \max(\alpha_j-(1+\epsilon)c_{il},0)	&\leq	&	(1+\epsilon)f_i			& \forall j \in \{1, \ldots, k\}& (3^{**}) 	\\
					%& \displaystyle \alpha_j,c_{ij}, f_i	&\geq&					0 & \forall j  \in \{1, \ldots, k\}								& (4^{**})\\
	%\end{array} 
%\]

\[
	\begin{array}{rclll}
		\multicolumn{5}{l}{ \textrm{{\bf  FacRev-Jain:}}}\\
		 \mbox{maximize}		 & \multicolumn{4}{l}{  z_k = \frac{\sum_{j=1}^{k} \alpha_j} {f_i + \sum_{j = 1}^{k}  c_{ij}}}\\
		
		 \mbox{subject to} \quad \displaystyle \alpha_j	&\leq	& \alpha_{j+1} 		& \forall j \in \{1, \ldots, k-1\}								& (1^*) \\
					 \displaystyle \alpha_j				&\leq	& \alpha_l + c_{ij} + c_{il}		& \forall j,l \in \{1, \ldots, k\} 	& (2^*)\\
					  \sum_{l=j}^{k} \max(\alpha_j-c_{il},0)	&\leq	&	f_i			& \forall j \in \{1, \ldots, k\}& (3^*) 	\\
					 \displaystyle \alpha_j,c_{ij}, f_i	&\geq&					0 & \forall j  \in \{1, \ldots, k\}								& (4^*)\\
	\end{array} 
\]

\[
	\begin{array}{rclll}
		\multicolumn{5}{l}{ \textrm{{\bf  FacRev-Dist:}}}\\
		 \mbox{maximize}		 & \multicolumn{4}{l}{  z_k = \frac{\sum_{j=1}^{k} \alpha_j} {f_i + \sum_{j = 1}^{k}  c_{ij}}}\\
		
		 \mbox{subject to} \quad \displaystyle \alpha_j	&\leq	& \alpha_{j+1} 		& \forall j \in \{1, \ldots, k-1\}								& (1^{**}) \\
					 \displaystyle \alpha_j				&\leq	& (1+\epsilon)(\alpha_l + c_{ij} + c_{il})		& \forall j,l \in \{1, \ldots, k\} 	& (2^{**})\\
					  \sum_{l=j}^{k} \max(\alpha_j-(1+\epsilon)c_{il},0)	&\leq	&	(1+\epsilon)f_i			& \forall j \in \{1, \ldots, k\}& (3^{**}) 	\\
					 \displaystyle \alpha_j,c_{ij}, f_i	&\geq&					0 & \forall j  \in \{1, \ldots, k\}								& (4^{**})\\
	\end{array} 
\]

Note that constraints $(1^*)$ \resp $(1^{**})$ are only used to sort the $\alpha_j$ values in order to ease the formulation of constraints $(3^*)$ \resp $(3^{**})$. The following theorem quantifies the correlation between our (\FRDist) and the original factor-revealing LPs (\FRJain).
\begin{theorem}
	For any $k \in \mathbb N$, the optimal solution to \FRDist is bounded by $(1+\epsilon)^2 1.861$, and thus the approximation factor of our algorithm is also $(1+\epsilon)^2 1.861$.
\end{theorem}
\begin{proof}

Jain \et showed that $z_k$ can be at most $1.861$ in \FRJain. In our modified version $z_k$ can be at most $(1+\epsilon)^2 1.861$: Fix a problem instance by setting the $c_{ij}$ variables and the $f_i$ variable to arbitrary values $\geq 0$ and consider the following two propositions.
\begin{proposition}\label{propositionLP1}  In an optimal solution $\alpha := (\alpha_1, \alpha_2, \ldots, \alpha_k)$ to \FRJain, $\forall j \in C$ at least one of the two properties is true:
\begin{enumerate}
	\item[(i)]	$\alpha_j = (f_i + \sum_{q \in Q_j} c_{iq})/|Q_j|$, where $Q_j := \{q| q \in C \wedge \alpha_j - c_{iq} \geq 0 \wedge \alpha_j \leq \alpha_q \}$
	\item[(ii)] $\exists l \in C: \alpha_l = (f_i + \sum_{q \in Q_l} c_{iq})/|Q_l|$ and $\alpha_j = \alpha_l + c_{ji} + c_{il}$ 
\end{enumerate}
\end{proposition}
\begin{proof}
	Note that the term $(f_i + \sum_{q \in Q_j} c_{iq})/|Q_j|$ is fixed. Since we deal with a maximization problem, which is bounded (see constraint $(3^*)$ of \FRJain), each $\alpha_j$ is bounded. Since $\alpha$ represents an optimal solution, for each $\alpha_j$ there must be at least one constraint that bounds $\alpha_j$ and is also tight. If this tight constraint is of the form $\sum_{l=j}^{k} \max(\alpha_j-c_{il},0)	\leq	f_i$, (i) holds.
	
Otherwise, there must be at least one tight constraint of the form $\alpha_j \leq \alpha_{k_1} + c_{ij} + c_{ik_1}$, i.e. $\exists k_1 \in C: \alpha_j = \alpha_{k_1} + c_{ij} + c_{ik_1}$. This argument can be analogously applied to $k_1$: either $\alpha_{k_1} = (f_i + \sum_{q \in Q_{k_1}} c_{iq})/|Q_{k_1}|$, or there exists a $k_2$ such that $\alpha_{k_1} = \alpha_{k_2} + c_{ik_1} + c_{ik_2}$. Like this, we can build a recursion for $\alpha_j = \alpha_{k_1}+c_{ij} +c_{ik_1}$, which terminates as soon as we reach a variable $\alpha_{k_t}$ with $\alpha_{k_t} = (f_i + \sum_{q \in Q_{k_t}} c_{iq})/|Q_{k_t}|$ and thus (i) holds for $k_t$.

Note that for every $\alpha_j$ there must be such a terminating sequence, since constraint $(3^*)$ is the only constraint which gives an absolute upper bound on the variables. Otherwise, the $\alpha_j$ would be unbounded. The recursion yields 
%This recursion terminates at $\alpha_{k_t}$ if $\alpha_{k_t} = (f_i + \sum_{q \in Q_{k_t}} c_{iq})/|Q_{k_t}|$ and thus (i) holds for $k_t$. For every $\alpha_j$ there must be at least one such sequence of $(\alpha_{k_1}, \alpha_{k_2}, \ldots, \alpha_{k_t})$ defined by the recursion, such that $\alpha_{k_t} = (f_i + \sum_{q \in Q_j} c_{iq})/|Q_j|$ (i.e. the recursion terminates). This is the case, because otherwise the $\alpha_j$ would be unbounded. Using the recursion, we can express $\alpha_j$ in the following way:
$\alpha_j =  \alpha_{k_t} + c_{ik_1} + c_{ik_1} +c_{ik_2} + c_{ik_2} + \ldots + c_{ik_{t-1}} + c_{ik_t}$, implying $\alpha_j \geq \alpha_{k_t} + c_{ij} + c_{ik_t}$. Since also $\alpha_j \leq \alpha_{k_t} + c_{ij} + c_{ik_t}$ (constraint $(2^*)$), we have $\alpha_j = \alpha_{k_t} + c_{ij} + c_{ik_t}$, proving (ii).  
\end{proof}

\begin{proposition} \label{propositionLP2}
	 Given an optimal solution $\alpha := (\alpha_1, \alpha_2, \ldots, \alpha_k )$ to \FRJain and an optimal solution $\alpha' := (\alpha'_1, \alpha'_2, \ldots, \alpha'_k )$ to \FRDist, we have $\forall j \in C: \alpha'_j \leq (1+\epsilon)^2\alpha_j$.	
\end {proposition}
\begin{proof}
	Since the constraints in \FRDist are a relaxation of the constraints of \FRJain, the solution $\alpha$ is also a feasible solution of \FRDist. Based on Proposition \ref{propositionLP1} we know that to any $\alpha_j$ at least one of the following cases applies:
	
\begin{enumerate} 	
	\item[(i)] $\alpha_j = (f_i + \sum_{q \in Q_j} c_{iq})/|Q_j|$ \ie the constraint $\sum_{l=j}^{k} \max(\alpha_j-c_{il},0)	\leq	f_i$ is tight. Notice that for $\alpha'_j$ satisfying the constraint $\sum_{l=j}^{k} \max(\alpha_j-(1+\epsilon)c_{il},0) \leq	(1+\epsilon)f_i$ we have, due to the set $\{c_{il} | \alpha_j > c_{il}\}$ containing the set $\{c_{il} | \alpha'_j > (1+\epsilon)c_{il}\}$,  $\alpha'_j \leq (1+\epsilon) \alpha_j$.
	
	\item[(ii)] According to Proposition \ref{propositionLP1} we know that $\exists l \in C: \alpha_l = (f_i + \sum_{q \in Q_l} c_{iq}/|Q_l|$ and $\alpha_j = \alpha_l + c_{ji} + c_{il}$. Constraint (2) of \emph{FR-Dist} guarantees that $\alpha'_j \leq (1+\epsilon) (\alpha'_l + c_{ji} + c_{il})$. Since we can apply case (i) for $l$, we have $\alpha'_l \leq (1+\epsilon)\alpha_l$, which yields $\alpha'_j \leq (1+\epsilon)((1+\epsilon)\alpha_l + c_{ji} + c_{il}) \leq (1+\epsilon)^2 (\alpha_l + c_{ji} + c_{il})$.
\end{enumerate}	
\end{proof}

Given any instance of the facility location problem, we know that the optimal solution to the \FRJain LP is bounded by $1.861$ and that, by Proposition \ref{propositionLP2}, the solution to \FRDist for the same instance is at most by a factor of $(1+\epsilon)^2$ larger. Since the problem instance was chosen arbitrarily, the theorem's statement follows.
%
%Since the \emph{FR-Dist} LP computes an upper bound for the approximation factor of our algorithm and since for any instance this bound is at most
%$(1+\epsilon)^2 1.861$, the theorem is proven.
\end{proof}

\section{Facility Selection Process} In this section we are only concerned with a single \emph{phase} of our algorithm (\ie the loop in which Algorithm 3 and 4 are invoked and which establishes Fact $1$). We show that it terminates after $\bigO{n^{3/4}\log(n)}$ rounds with high probability.

To simplify notation, $F$ and $C$ do not refer to all facilities \resp all clients (as was the case before). Instead, let $F$ represent the set of all facilities that are in the \emph{currently-paid} state and $C$ denote the set of all \emph{not connected} clients contributing to a \emph{currently-paid} facility.  All the other clients and facilities effectively do not take part in the execution of Algorithm 3 and 4 in the considered phase and thus can be ignored. We consider the bipartite Graph $G=(F\cup C,E)$. There is an edge in $E$ between $i \in F$	and $j \in C$ if and only if client $j$ contributes to $i$ (i.e. $\alpha_j - c_{ij} > 0$). Furthermore, we assume that there are no isolated nodes (\ie $\deg(v)\geq1$ for all $v\in F\cup C$). Let $n = \abs{F}+\abs{C}$. By $G_{F}\coloneqq(F,E_{F})$ we denote the $\emph{Facility Graph}$, where there is an edge $\{i,i'\}$ between two facilities $i,i'\in F$ if and only if they share a common client in $G$. We use $\deg(\cdot)$ and $N(\cdot)$ to denote the degree and neighborhood of a node in $G$ respectively. If the neighborhood includes the node itself, we write $N^+(\cdot)$. Similarly, we use $\deg_{F}(\cdot)$, $N_{F}(\cdot)$, and $N_{F}^+(\cdot)$ to denote the corresponding properties in $G_{F}$.

 In each execution of Algorithm 3 and 4 the set $F$ of \emph{currently-paid} facilities shrinks (they change their status to \emph{closed} or \emph{open}). Also, with every facility changing its state from \emph{currently-paid} to \emph{open} clients that are connected to those newly opened facilities are removed from $C$. Removing these clients form $C$ can potentially cause \emph{currently-paid} facilities to change their state to \emph{closed} (they lose contributing clients and their opening costs are not paid for anymore). Since it cannot be guaranteed that a \emph{currently-paid} facility which loses a contributing client is not longer fully paid for, applying a distributed MIS (maximal independent set) algorithm (\eg Luby's MIS algorithm \cite{luby}) on $G_{F}$ does not help to solve our problem: A facility that loses a contributing client (\ie a facility that is not part of the MIS) might still be fully paid for. If such a facility is not opened in the current phase (\ie, (iii) of Fact 1 does not hold), clients contributing to it will raise their $\alpha_j$ values in the next phase, which in turn could violate the constraints of the \FRDist LP. Since shrinking $F$ by removing \emph{currently-paid} facilities that are not paid for only decreases the runtime of Algorithm 3 and 4, we will assume (in order to consider the worst case) that facilities never change their status from \emph{currently-paid} to \emph{closed} as the result of clients changing their state to \emph{connected} (\ie, losing contributing clients). 

\paragraph{A simplified problem description.} In order to alleviate describing the shrinking process of the set of \emph{currently-paid} facilities, we simplify the analysis of Algorithm 3 and 4 by omitting the information about the values of $c_{ij}$ and $f_i$ and simulating the distributed execution of the algorithms by the algorithm \ouralg. Its input is the bipartite graph $G$ and it possesses global knowledge.

\begin{algorithm}[ht!]
\caption{\ouralg}
\label{alg:random}
\begin{algorithmic}[1]
\WHILE{ ($F \neq \emptyset$)}
	\STATE $I := \emptyset$ (\emph{Subsequently, determine $I\subseteq F$ such that $\{i,i'\}\notin E_{F}$ for all $i,i'\in I$)})\\	
	\FORALL{ $ i \in F:$} 
		\STATE Generate a uniformly distributed random value $r_i\in[0,1]$\\
	\IF{($r_i>\max_{i'\in N_{F}(i)}r_{i'}$)}
		\STATE Add $i$ to $I$\\
	\ENDIF
	\ENDFOR	
	\STATE Eliminate $N^+(I)$ from $G$\\
\ENDWHILE
\end{algorithmic}
\end{algorithm}

As the execution of \ouralg progresses, $F$, $C$, $E$, and $n$ change.
We denote these sets and $n$ in round $t$ as $F_t$, $C_t$, $E_t$ and $n_t$.
However, most of the time we consider only the effect of a single iteration of the while loop on the graph.
In these cases, we omit the time parameter and simply use $F$, $C$, $E$, and $n$ to refer to the corresponding values in the considered iteration.

\paragraph{Time required to shrink the bipartite graph.} This paragraph analyzes the time required to remove all nodes from $G$. It is easy to see that a fixed facility $i\in F$ is selected with probability $1/(\deg_{F}(i)+1)$ by this strategy. First, we show that, in expectation, \ouralg removes at least $\abs{E}/\abs{F}$ clients and at least $\max\{\abs{F},\abs{E}/\abs{F}\}$ edges in a single iteration of the while loop.

\begin{lemma}\label{prop:lower_bound_on_expected_eliminated_clients}
The expected number of clients removed by \ouralg in a single iteration of the while loop is at least $\abs{E}/\abs{F}$.
\end{lemma}
\begin{proof}
For a facility $i\in F$, let $X_i$ denote a random variable counting the number of clients $j\in {C}$ removed because of $i$ being added to $I$ or not.
That is,
\begin{equation*}
X_i=
\begin{cases}
\deg(i) & \text{if } i\in I\\
0       & \text{otherwise}\enspace.
\end{cases}
\end{equation*}
Since $I$ is an independent set in $G_{F}$, the total number of clients removed in one iteration can be described by the random variable $X\coloneqq\sum_{i\in F}X_i$.
Using the inequality $\deg_{F}(i)+1\leq \abs{F}$, we compute:
\begin{equation*}
E(X)=\sum_{i\in F}E(X_i)=\sum_{i\in F}\frac{\deg(i)}{\deg_{F}(i)+1}\geq\sum_{i\in {F}}\frac{\deg(i)}{\abs{F}}=\frac{\abs{E}}{\abs{F}}\enspace.
\end{equation*}
\end{proof}

\begin{lemma}\label{prop:lower_bound_on_expected_eliminated_edges}
The expected number of edges removed by \ouralg in a single iteration of the while loop is at least $\max\{\abs{F},\abs{E}/\abs{F}\}$.
\end{lemma}
\begin{proof}
First note that, since by Lemma~\ref{prop:lower_bound_on_expected_eliminated_clients} the expected number of clients eliminated is at least $\abs{E}/\abs{F}$ and there are no isolated clients, the expected number of edges removed is at least $\abs{E}/\abs{F}$. It remains to show that the expected number of edges removed is at least $\abs{F}$. For a facility $i\in F$, let us define random variables $Z_i$ counting the number of edges removed because of $i$ being added to $I$ or not.
That is,
\begin{equation*}
Z_i=
\begin{cases}
\sum_{j\in N(i)}\deg(j) & \text{if } i\in I\\
0                       & \text{otherwise}\enspace.
\end{cases}
\end{equation*}
Since $I$ is an independent set in $G_{F}$, the total number of edges removed by strategy \ouralg can be described by the random variable $Z$ defined as $Z\coloneqq\sum_{i\in F}Z_i$.
Using the inequality $\deg_{F}(i)+1\leq \sum_{j\in N(i)}\deg(j)$, we compute:
\begin{equation*}
E(Z)=\sum_{i\in F}E(Z_i)=\sum_{i\in F}\frac{\sum_{j\in N(i)}\deg(j)}{\deg_{F}(i)+1}\geq\sum_{i\in F}\frac{\sum_{j\in N(i)}\deg(j)}{\sum_{j\in N(i)}\deg(j)}=\abs{F}\enspace.
\end{equation*}
\end{proof}

As an corollary, we get that, in expectation, the number of edges is reduced by at least its square root:

\begin{corollary}\label{cor:lower_bound_on_expected_eliminated_edges}
The expected number of edges removed by \ouralg in a single iteration of the while loop is at least $\sqrt{\abs{E}}$.
\end{corollary}
\begin{proof}
	Lemma~\ref{prop:lower_bound_on_expected_eliminated_edges} and the equality $\abs{E}=\abs{F}\cdot\frac{\abs{E}}{\abs{F}}$ imply $\max\{\abs{F},\frac{\abs{E}}{\abs{F}}\}\geq\sqrt{\abs{E}}$.
\end{proof}

\begin{lemma}
If the number of edges is greater or equal $n\sqrt{n}$, the expected number of \emph{clients} removed by \ouralg in one iteration of the while loop is at least $\sqrt{\abs{C}}$.
If it is smaller or equal $n\sqrt{n}$, the expected number of removed \emph{edges} is at least $\sqrt{n\sqrt{n}}$.
\end{lemma}
\begin{proof}
For the first statement, we can bound the expected number of removed clients with the help of Lemma \ref{prop:lower_bound_on_expected_eliminated_clients} as follows:
\begin{equation*}
	\frac{\abs{E}}{\abs{F}} \geq \frac{n\sqrt{n}}{\abs{F}} \geq \frac{n\sqrt{n}}{n} = \sqrt{n} \geq \sqrt{\abs{C}}
\end{equation*}	
The second statement follows immediately from Corollary~\ref{cor:lower_bound_on_expected_eliminated_edges}.
\end{proof}

We call a round $t$ \emph{edge-heavy}, if in this round the number of edges is at least $n_t\sqrt{n}$.
Otherwise, we refer to it as \emph{edge-light}. 

\begin{lemma}\label{client_constant_fraction}
For an arbitrary round $t$, consider the effect of several iterations of \ouralg's while loop starting at round $t$.
\begin{enumerate}[(a)]
\item After $\frac{1}{2}\sqrt{n_t}$ (not necessarily consecutive) \emph{edge-heavy} rounds, at least $\frac{1}{2\sqrt{2}}\abs{C_t}$ clients are removed in expectation.
\item After $\frac{1}{2}\sqrt{n_t\sqrt{n_t}}$ (not necessarily consecutive) \emph{edge-light} rounds, at least $\frac{1}{2\sqrt{2}}\abs{E_t}$ edges are removed in expectation.
\end{enumerate}
\end{lemma}
\begin{proof}
	We prove only the first statement, since the second can be proven analogously. Let $Y$ denote a random variable counting the number of clients removed during $\frac{1}{2}\sqrt{n_t}$ edge-heavy rounds. If $Y<\abs{C_t}/2$, the number of clients in each of the considered rounds was at least $\abs{C_t}/2$. Thus, by Lemma~\ref{prop:lower_bound_on_expected_eliminated_clients}, in each of these $\frac{1}{2}\sqrt{n_t}$ rounds at least $\sqrt{\abs{C_t}/2}$ clients were removed in expectation. If $Y\geq\abs{C_t}/2$, the number of removed clients is trivially lower bounded by $\abs{C_t}/2$. We get:
	\begin{equation*}
	\begin{split}
	E[Y]&\geq\frac{\abs{C_t}}{2}\cdot Pr\left(Y\geq\frac{\abs{C_t}}{2}\right)+\frac{1}{2}\sqrt{n_t}\cdot\sqrt{\frac{\abs{C_t}}{2}}\cdot Pr\left(Y<\frac{\abs{C_t}}{2}\right)\\
	&\geq\frac{\abs{C_t}}{2\sqrt{2}}\cdot Pr\left(Y\geq\frac{\abs{C_t}}{2}\right)+\frac{\abs{C_t}}{2\sqrt{2}}\cdot Pr\left(Y<\frac{\abs{C_t}}{2}\right) =\frac{\abs{C_t}}{2\sqrt{2}}
	\enspace.
	\end{split}
	\end{equation*}
\end{proof}

\begin{lemma}\label{client_edge_constant_probability}
For an arbitrary round $t$, consider the effect of several iterations of \ouralg's while loop starting at round $t$.
\begin{enumerate}[(a)]
\item\label{client_constant_probability} The probability of removing at least $\frac{1}{4\sqrt{2}}\abs{C_t}$ clients after $\frac{1}{2}\sqrt{n_t}$ \emph{edge-heavy} rounds is at least a constant $>0$.
\item\label{edge_constant_probability} The probability of removing at least $\frac{1}{4\sqrt{2}}\abs{E_t}$ edges after $\frac{1}{2}\sqrt{n_t\sqrt{n_t}}$ \emph{edge-light} rounds is at least a constant $>0$.
\end{enumerate}
\end{lemma}
\begin{proof}
We prove only the first statement.
The second can be proven analogously.
	
By Lemma \ref{client_constant_fraction} $\frac{1}{2\sqrt{2}}\abs{C_t}$ clients are removed after $\frac{1}{2}\sqrt{n_t}$ \emph{edge-heavy} rounds in expectation. Let $X$ define a random variable counting the number of clients removed after $\frac{1}{2}\sqrt{n_t}$ \emph{edge-heavy} rounds. Now define $\mu$ to be the probability that more than $\frac{1}{4\sqrt{2}}\abs{C_t}$ clients are removed.
We can bound $E[X]$ in the following way:
\begin{equation*}
	\frac{1}{2\sqrt{2}}\abs{C_t} \leq E[X] \leq \mu n + (1-\mu) \frac{1}{4\sqrt{2}}\abs{C_t}
\end{equation*} 
Solving this inequality for $\mu$ yields $\mu \geq c$, where $c$ is a positive constant.
\end{proof}

We are now ready to prove our main theorem.

\begin{theorem}
Algorithm \ouralg terminates with high probability after $\bigO{n^{\frac{3}{4}}\log(n)}$ iterations of the while loop.
\end{theorem} 
\begin{proof}
By Lemma~\ref{client_edge_constant_probability}, starting at iteration $t$ we lose at least a constant fraction of either the clients or the edges with a non-zero, constant probability after at most $\frac{1}{2}\sqrt{n_t}+\frac{1}{2}\sqrt{n_t\sqrt{n_t}}\leq n_t^{3/4}$ rounds.
Partition the number of iterations done by \ouralg into phases of $n^{3/4}$ rounds.
Let $X_i$ denote a binary random variable indicating whether the $i$-th phase was \emph{good} (we lost at least the above mentioned constant fraction of either clients or edges).
We need about $2\log n$ good rounds to ensure that the problem size has become $\bigO{1}$.
Since we have $Pr(X_i=1)\geq c$ for some positive constant $c$, we get for $X\coloneqq\sum_{i=1}^{\frac{1}{c}2\log n}X_i$:
\begin{equation*}
E[X]\geq2\log(n)
\enspace.
\end{equation*}
That is, in expectation it takes about $\frac{1}{c}\cdot2\log(n)$ phases until \ouralg terminates.
Applying a standard Chernoff bound yields:
\begin{equation*}
Pr\left(\abs{X-E[X]}>\frac{E[X]}{2}\right)< 2^{-4E[X]/3}
\enspace.
\end{equation*}
In other words, \ouralg terminates with high probability after at most $\frac{1}{c}2\log(n)$ steps.
\end{proof}

\section{Conclusion and Outlook} We presented a parallel execution of a greedy sequential algorithm for the facility location problem and showed that we can preserve its approximation factor. However, there are other sequential facility location algorithms based on the greedy approach introduced by Jain \et that yield (better) approximation factors of 1.61 \cite{FLgreedy} and 1.52 \cite{FLMahdian152}. It might be possible to execute them in the $\Congest$ model with a sublinear running time using the technique we introduced in this paper.
 
When our algorithm chooses which facilities should be opened, it bases its decision solely on the information provided by the bipartite graph $G$ used in Section 3 and all other information (client contribution, distance, facility opening costs ect.) is abstracted away (Section 3 provides reasons for this abstraction). At this level of abstraction the lower bound of solving the problem is $\sqrt{n}$: consider a bipartite graph with $\sqrt{n}$ facilities with degree of $\sqrt{n}-1$, and $\binom{\sqrt{n}}{2}$ clients with degree $2$. Edges are chosen in such a way that any two facilities are in hop distance of 2 of each other. It would be interesting to know in what way the entire available information can be used to improve the runtime and whether a poly-logarithmic runtime is possible with the same approximation factor.

 \newpage
{\footnotesize
\bibliographystyle{alpha}
\bibliography{references}}

\newpage
\appendix

\section{Proofs for Lemma 1 and 2}
\begin{proof}[Proof of Lemma~\ref{lem:metric_lp_property}]
	If $\alpha_j/(1+\epsilon) \leq \alpha_{j'}$ the lemma obviously holds. Thus, let us assume the opposite. Let $t$ be the phase when our algorithm connects client $j'$ to some facility $i'$ and $s$ the phase $i'$ was opened. Facility $i'$ was either opened before $j'$ was connected to it ($s < t$) or opened exactly in the same phase client $j'$ was connected to it ($s = t$). Since $\alpha_{j} > \alpha_{j'}$, we know that client $j$ was connected in phase $u \geq t+1$.  At time $s \leq t$ the contribution of client $j$ was not enough to cover the distance between $j$ and $i'$, since otherwise $j$ would be connected in phase $s$ with $i'$. This means that once the contribution of $j$ reaches $c_{i'j}$, client $j$ will be connected to $i'$. Since the contribution is increased by multiplying the current $\alpha$ value with $(1+\epsilon)$, it can not become greater than $(1+\epsilon) c_{i'j}$, hence $\alpha_j \leq (1+\epsilon) c_{i'j}$. Thanks to the triangle inequality we have $\alpha_j / (1+ \epsilon) \leq c_{i'j} \leq c_{i'j'} + c_{ij'} + c_{ij} \leq \alpha_{j'} + c_{ij'} + c_{ij}$.
\end{proof}

\begin{proof}[Proof of Lemma~\ref{lem:facilities_not_overpaied_property}]
	To prove by contradiction, we assume that the inequality does not hold. Thus, we have $\sum_{l=j}^{k} \max(\alpha_j - (1+\epsilon)c_{il}, 0) > (1+\epsilon)f_i$. Since we ordered the clients according to their $\alpha$ value, we know that for $l \geq j$ we have $\alpha_l \geq \alpha_j$. Let $t$ be the phase client $j$ was connected to a facility. By the assumption facility $i$ is fully paid for (i.e\text{.} $i$'s current payment is $\geq f_i$)  in phase $s < t$. There must be at least one client $l$, with $j \leq l \leq k$ and $\alpha_j - (1+\epsilon)c_{il} > 0$ which was connected to $i$ in a phase $q < t$. For this client $l$, we have $\alpha_l < \alpha_j$, since the algorithm will stop increasing $\alpha_l$ when $l$ is connected to a fully paid facility (see Fact 1). 
\end{proof}

\end{document}